\documentclass[runningheads]{llncs}
\usepackage{paper}

\input{anc/data/saved.tex}

\title{Fairest Neighbors}
\subtitle{Tradeoffs Between Metric Queries}

\author{Magnus Lie Hetland\orcidID{0000-0003-4204-2017} \and
Halvard Hummel\orcidID{0000-0001-5691-8177}}
\institute{Norwegian University of Science and Technology, Trondheim, Norway
\email{\{mlh,halvard.hummel\}@ntnu.no}}

\authorrunning{M. L. Hetland and H. Hummel}

\begin{document}
\maketitle

\begin{abstract}
Metric search commonly involves finding objects similar to a given sample
object.
We explore a generalization, where the desired result is a fair tradeoff between
multiple query objects. This builds on previous results on complex queries, such
as linear combinations. We instead use measures of inequality, like ordered
weighted averages, and query existing index structures to find objects that
minimize these. We compare our method empirically to linear scan and a post hoc
combination of individual queries, and demonstrate a considerable speedup.
\keywords{Metric indexing \and Multicriteria decisions \and Fairness}
\end{abstract}

\section{Introduction}

From the early days,
indexing metric spaces has mainly been in service of straightforward
similarity search: Given some query object $q$, find other objects $o$ for which
the distance $d(q,o)$ is low---either all points within some search radius, or a
certain number of the nearest neighbors. Alternative forms of search have been
explored, certainly.
Of particular interest to us is using multiple query objects $q_i$, without
restricting the indexing methods used. That is, we wish to take any existing
metric index, already constructed, and execute a combination query on it. Such a
query may be specified directly by the user, or it may be a form of interactive
refinement. A user first performs a query using a single object. Then she
indicates which of the returned objects are most relevant (possibly to varying
degrees), and these are then used as a second, combined query. The result should
ideally be a tradeoff between the query objects. In particular, we wish to
ensure that it is a \emph{fair} tradeoff, borrowing measures of fairness from
the field of multicriteria decision making.

\paragraph{Our contributions.}
In this short paper, we introduce the idea of \emph{fairest neighbors} ($k$FN),
i.e., items that are close to multiple query objects at once, as measured by
some kind of fairness measure. For example, if we are looking for a centaur,
using a human and a horse, a simple linear combination will not do, as the best
results are then just as likely to be similar to just the human or just the
horse; only a fair combination would give us what we want.
We formulate such queries in the context of the complex queries of
\citet{Ciaccia:1998}, but extend the formalism by applying linear ambit
overlap~\citep{Hetland:2019} to ordered weighted averages (OWA) and weighted
OWA, for improved bounds. The resulting queries may be resolved using existing
metric index structures \emph{without modification}. We perform preliminary
experimental feasibility tests, showing that such combined $k$FN queries
outperform both linear scan and using multiple separate $k$NN queries.

\paragraph{Related work.}
Others have studied metric search with multiple simultaneous criteria.
As discussed by \citet{Ciaccia:1998}, Fagin's $\mathcal{A}_0$ algorithm also
resolves complex queries, but makes additional assumptions about synchronized
independent subsystems. \Citet{Bustos:2006} study a superficially similar
problem that involves a linear combination of multiple metrics, while still
using a single query object. More closely related are \emph{metric skylines},
which are essentially Pareto frontiers in pivot space~\citep{Chen:2008}. These
result sets will be diverse, and will tend to include both fair and unfair
solutions. Our approach moves beyond non-Pareto-dominated solutions to
non-Lorenz-dominated solutions~\citep[cf.][]{Gonzales:2020}.

\section{Complex Queries as Multicriteria Decisions}
\label{sec:complex}

In \citeyear{Ciaccia:1998}, \citeauthor{Ciaccia:1998} introduced a formalism for
dealing with what they called \emph{complex queries} in metric indexes---queries
involving multiple query objects, along with some domain-specific query
\emph{language}, specifying which objects are relevant and which
aren't~\citep{Ciaccia:1998}. Part of their formalism involves mapping distances
to similarity measures, which are then constrained by some query predicate;
however, the core ideas apply equally well to distances directly. A central
insight is that \emph{monotone} predicates may be used not just to detect
whether an \emph{object} is relevant, but also whether certain \emph{regions}
might contain relevant objects.

Let $x=[d(q_i,o)]_{i=1}^m$ be a vector of
distances between query objects $q_i$ and some potentially relevant object $o$.
Relevance is then defined by some predicate on these distances, $\Pred(x)$. This
predicate is \emph{monotone} if for all $x\leq y$ we have that $\Pred(y)$
implies $\Pred(x)$. That is, if we start with the distance vector of a relevant
object, and we reduce one or more of the distances, the resulting vector should
\emph{also} be judged as relevant. In this case, using \emph{lower bounds} for
the individual distances is safe (i.e., it will not cause false negatives). So,
for example, if we know that $o$ is in a ball with center $c$ and radius $r$, we
can safely replace $[d(q_i,o)]_i$ with $[d(q_i,c)-r]_i$ and apply $\Pred$ to
determine whether or not to examine the ball.
Using this approach, one can find the $k$ best objects by maintaining a steadily
shrinking search radius encompassing the $k$ best candidates found so far, just
as one would for $k$NN.
The idea is illustrated in \cref{fig:lowerbound}: The vector $x-r$ of lower
bounds corresponds to the lower left corner of the square enveloping the region
in pivot space (indicated by a `$+$' in the right-hand subfigure). A monotone
query and a ball region may overlap only if this lower left point is inside the
query definition in this space~\citep[cf.][]{Hetland:2019}. Similarly, we may
conclude that the region is \emph{entirely inside} the query (and thus return
all its objects without further examination) if the upper right corner ($x+r$)
satisfies the query predicate.

\begin{figure}[t]
\input{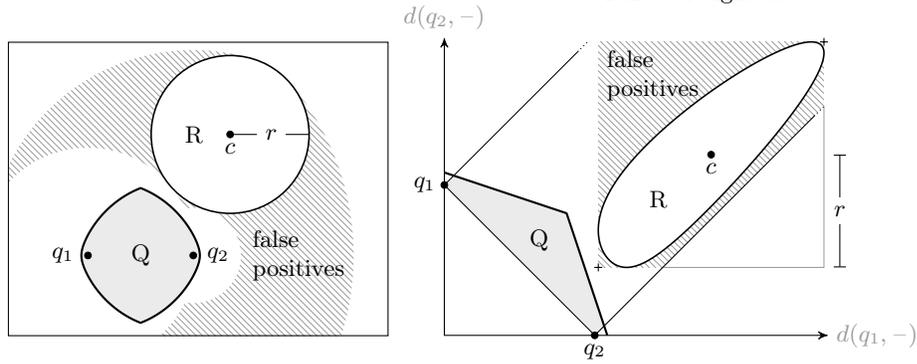}
\hfill
\begin{tikzpicture}
\begin{scope}
\clip (0,0) rectangle (\centerfocalx+\radius,\centerfocaly+\radius);
\node at (0.5*\centerfocalx+0.5*\radius,0.5*\centerfocaly+0.5*\radius) {%
\begin{tikzpicture}[scale=.7]

    \def\hatchsep{1.7143 pt}

    \begin{scope}[overlay, even odd rule]

    \clip (-1,0)
        circle[radius=\centerfocalx-\radius]
        circle[radius=\centerfocalx+\radius]
        ;
    \clip (+1,0)
        circle[radius=\centerfocaly-\radius]
        circle[radius=\centerfocaly+\radius]
        ;

    \foreach \x in {0,1,...,80} {
        \draw[overlay, 
        very thin, secondary color]
        (-4,1) ++(\x*\hatchsep,\x*\hatchsep)
         -- +(6,-6);
    }

    \end{scope}
    \draw[thick, fill=light color] \giniambit;
    \draw[fill=white, semithick] (\centerx,\centery) circle[radius=\radius];

    \draw
        (-1,0) node[point] {} node[left=2pt] {$q_1$}
        (+1,0) node[point] {} node[right=2pt] {$q_2$}
        (\centerx,\centery) node[point] (p) {} node[below] {$c$}
        node[left=7pt] {$\mathrm R$}
        ;

    \draw (p) -- node[fill=white, midway] {$r$} +(\radius,0);

    \draw (0,0) node {$\mathrm Q$};

    \draw (3,0) node {\begin{varwidth}{5cm}
        false\\
        positives
    \end{varwidth}};

\end{tikzpicture}};
\end{scope}
\draw (0,0) rectangle (\centerfocalx+\radius,\centerfocaly+\radius);
\end{tikzpicture}
\hfill
\begin{tikzpicture}[scale=1]

    \begin{scope}
        \clip[overlay]
        (\centerfocaly+\radius-\sep,\centerfocaly+\radius)
        -- (0,\sep) -- (\sep,0) --
        (\centerfocalx+\radius,\centerfocalx+\radius-\sep)
        --
        (\centerfocalx+\radius,\centerfocalx+\radius)
        ;
    \clip[overlay]
        (\centerfocalx-\radius,\centerfocaly-\radius)
        rectangle
        +(2*\radius,2*\radius)
        ;
    \foreach \x in {0,1,...,71} {
        \draw[very thin, secondary color]
        (\centerfocalx-2*\radius,\centerfocaly) ++(\x*1.2 pt,\x*1.2 pt)
         -- +(2*\radius,-2*\radius);
    }
    \end{scope}

    \begin{scope}[even odd rule]
        \clip[overlay] (0,0) rectangle (10, 10)
        (\centerfocaly+10-\sep,\centerfocaly+10)
        -- (0,\sep) -- (\sep,0) --
        (\centerfocalx+10,\centerfocalx+10-\sep)
        --
        (\centerfocalx+10,\centerfocalx+10)
        ;
        \draw[very thin, secondary color]
            (\centerfocalx-\radius,\centerfocaly-\radius)
            rectangle
            (\centerfocalx+\radius,\centerfocaly+\radius)
            ;
    \end{scope}

    \draw[overlay, thin]
        (\centerfocalx-\radius,\centerfocaly-\radius)
            +(-1.5pt,0) -- +(1.5pt,0)
            +(0,-1.5pt) -- +(0,1.5pt)
        (\centerfocalx+\radius,\centerfocaly+\radius)
            +(-1.5pt,0) -- +(1.5pt,0)
            +(0,-1.5pt) -- +(0,1.5pt)
        ;

    \begin{scope}[overlay, clip]

        \clip
            (\sep,2*\sep) -- (0,\sep) -- (\sep,0)
            -- (2*\sep,\sep) -- (2*\sep,2*\sep)
            -- cycle
            ;

        \fill[fill=light color] \ginifocal -- (0,0) -- cycle;

    \end{scope}

    \draw[thick] \ginifocal;

    \draw[semithick, fill=white] \regionfocal -- cycle;

    \draw (1.26,1.26) node {$\mathrm Q$};

    \draw (\centerfocalx, \centerfocaly) node[point] (p) {}
        node[below] {$c$}
        +(220:26pt) node {$\mathrm R$}
        ;

    \draw[...-...]
        (\centerfocaly+\radius-\sep,\centerfocaly+\radius)
        -- (0,\sep) -- (\sep,0) --
        (\centerfocalx+\radius,\centerfocalx+\radius-\sep)
        ;

    \draw[overlay, <->]
        (0,\centerfocaly+\radius+.05) node[black!40, above] {$d(q_2,\blank)$} -- (0,0) --
        (\centerfocalx+\radius+.05,0) node[black!40, right] {$d(q_1,\blank)$}
        ;

        \draw[overlay]
        (0,\sep) node[point] {} node[left] {$q_1$}
        (\sep,0) node[point] {} node[below] {$q_2$}
        ;

        \draw (\centerfocalx-\radius,\centerfocaly+\radius)
        node[anchor=north west]
        {\begin{varwidth}{5cm}
        false\\
        positives
        \end{varwidth}};

        \draw[overlay, |-|] (p) ++(\radius,0) ++(6pt,0) -- node[midway,
        fill=white] {$r$} +(0,-\radius);

\end{tikzpicture}
\hfill\mbox{}
\caption{A complex query $\mathrm Q$ with two query objects $q_i$ and a ball
region $\mathrm R$. The right-hand subfigure shows the query and region in pivot
space, where the two axes correspond to distances from the two query objects.
The hatched areas could potentially fall within the ball $\mathrm R$, depending
on the nature of the metric; what remains outside $\mathrm R$ constitutes
potential false positives. (\Citeauthor{Hetland:2019} discusses these concepts
in depth~\citep{Hetland:2019}.)}
\label{fig:lowerbound}
\end{figure}

Two of the query types discussed explicitly by \citeauthor{Ciaccia:1998} are
based on fuzzy logic, and one uses a weighted sum. These permit indicating
degrees of relevance for the various query objects $q_i$, but may have
\emph{many equally good solutions}, with vastly different properties. What can
be done if we wish to enforce some form of actual tradeoff? Consider a query
predicate of the form $f(x)\leq s$. That is, we apply some monotone function $f$
to the vector $x$ of distance $d(q_i,o)$ and are only interested in objects $o$
for which $f(x)$ falls below some search radius $s$. Different monotone
functions $f$ may yield very different query regions:
\vspace{1ex}
\begin{center}
\hfill
\begin{tikzpicture}[font=\footnotesize]
    \draw
        (-.6,0) node[point] (L) {}
        (+.6,0) node[point] (R) {}
    ;
    \draw[overlay]
        (L) circle[radius=.25]
        (R) circle[radius=.25]
    ;
    \path (0,0) circle[radius=.25];
    \draw (0,-.6) node[inner sep=0] {\smash{$\min(x_1,x_2)$}};
\end{tikzpicture}
\hfill
\begin{tikzpicture}[font=\footnotesize]
    \draw
        (-.6,0) node[point] (L) {}
        (+.6,0) node[point] (R) {}
    ;
    \draw (L) edge (R);
    \path (0,0) circle[radius=.25];
    \draw (0,-.6) node[inner sep=0] {\smash{$x_1+x_2$}};
\end{tikzpicture}
\hfill
\begin{tikzpicture}[font=\footnotesize]
    \draw
        (-.6,0) node[point] (L) {}
        (+.6,0) node[point] (R) {}
    ;
    \draw (0,0) circle[radius=.25];
    \draw (0,-.6) node[inner sep=0] {\smash{$x_1^2+x_2^2$}};
\end{tikzpicture}
\hfill\mbox{}\vspace{-.75ex}
\end{center}
\vspace{.75ex}
Minimum (corresponding to maximum, or standard fuzzy disjunction, in the
similarity formalism of \citeauthor{Ciaccia:1998}) produces results that are
close to one or the other of the two query points, but not both. A sum gives us
points that can lie anywhere between the two (in general within an ellipsoid). A
sum of positive powers, however, produces items that are \emph{between} the
query points---ideally in the middle (i.e., in their metric midset). This is the
kind of query we want.

Using sums of powers to characterize tradeoffs is a common approach
in cardinal welfarism,
and it is one of a broader class of aggregation functions used in multicriteria
decision making~\citep{Gonzales:2020}.\footnote{Though \citeauthor{Ciaccia:1998}
do not directly address fairness or tradeoffs, their \emph{standard} and
\emph{algebraic fuzzy conjunctions}, correspond to the \emph{maximin} and
\emph{Nash welfare} fairness measures, respectively, if applied, in isolation,
to similarities~\citep{Ciaccia:1998}.} These are all generally monotonically
increasing, with the optimum found for some fair tradeoff between their
arguments. Applied to individual query distances $d(q_i,o)$, our measure will
of course need to be \emph{minimized}, and so must be an \emph{unfairness}
measure, rather than a \emph{fairness} measure.
In the following, we will focus on \emph{ordered weighted average} (OWA), and
its generalization, \emph{weighted} OWA. The OWA of some vector $x$ is based on
a weighting of the elements of $x$, just like a weighted average, except that
the weights are applied based on the rank of each element $x_i$. Given
a weight vector $w\geq 0$, summing to~1, the OWA of $x$ is
$wx^\uparrow$, where, $x^\uparrow$ is a sorted version of $x$.
As discussed in \cref{sec:linear} (in a more general setting), by ensuring that
$w$ is also sorted, we get an unfairness measure. Our overlap check with an
$r$-ball, using the complex query formalism, becomes:
\begin{equation}
    f(x-r) \leq s
    \iff
    wx^\uparrow - r \leq s\,.
\end{equation}
For some structures, such as VP-trees~\citep{Yianilos:1993},
LC~\citep{Chavez:2005}
and HC~\citep{Fredriksson:2007}, we also
need to determine whether the query is \emph{entirely
inside} a ball region---or, equivalently, whether it intersects with the
complement of the ball. Our lower bound on each distance between the query and
the outside is $r-x_i$, and using monotonicity again, we get the criterion
$r-wx^\uparrow < s$. If, however, we do not treat the query as a black-box
monotone function, we can, as described in the following section, get the
stronger criterion $r-wx^\downarrow < s$, where $x^\downarrow$ is $x$ sorted in
descending order. The difference between these two bounds can be
\emph{arbitrarily large}, even for just two query objects. The complemented ball
is also just a particularly simple linear ambit with negative
coefficients~\citep{Hetland:2019};
the situation is similar for other such regions.

\section{Ordered Weighted Averages and Linear Ambits}
\label{sec:linear}

It is possible to construct a weighted generalization of OWA, called
\emph{weighted} OWA (WOWA), where
some individuals (i.e., query objects) get preferential treatment
when determining a tradeoff~\citep{Torra:1997}.
The following definition is given by \citet{Gonzales:2020}.

\begin{definition}
    \label{def:wowa}
    Let $p = [p_1, \dots, p_m]$ and $w = [w_1, \dots, w_m]$ be weighting
    vectors, where $p_i, w_i \in [0, 1]$ and $\sum_{i = 1}^m p_i = \sum_{i =
    1}^m w_i = 1$. The \emph{weighted ordered weighted average (WOWA)} of
    a vector $x \in \mathbb{R}^m$ with respect to $p$ and $w$ is defined by:
    \begin{equation}
        \WOWA(x; p, w) = \sum_{i = 1}^m \left[\varphi\!\left(\sum_{k = i}^m
        p_{\sigma(k)\!}\!\right) - \varphi\!\left(\smashoperator[r]{\sum_{k = i +
        1}^m} p_{\sigma(k)\!}\!\right)\right]x_{\sigma(i)}\;,
    \end{equation}
    where $\sigma$ is a permutation of $x$ in increasing order and $\varphi :
    [0, 1] \to [0, 1]$ is defined by linear interpolation between
    values $\varphi(i / m) = \sum_{k = 1}^i w_{m - k + 1}$ and $\varphi(0) = 0$.
\end{definition}
With decreasing weights, WOWA is a fairness measure. This works well for
similarities, but as discussed, for distances we need need \emph{un}fairness.
One way of achieving this is to use an \emph{increasing} weight vector.
This makes intuitive sense, and for similarities $s(u, v) = 1 - d(u, v)$, as
used by \citet{Ciaccia:1998}, we can show that the least unfair distance
tradeoff is exactly the fairest similarity tradeoff.%
\footnote{Note that, following \citeauthor{Ciaccia:1998}, we assume $s(u,
v)\in [0, 1]$, which requires a bounded metric, with $d(u, v)\in [0, 1]$.}
\begin{proposition}
    \label{prop:equiv}
    Let $p$, $w$ and $w'$ be WOWA weighting vectors, with $w'_i = w_{m - i + 1}$
    for all $i \in \{1, \dots, m\}$. For any $x \in [0, 1]^m$, we have that:
    \begin{equation}\label{eq:prop-wowa}
        \WOWA(x; p, w') = 1 - \WOWA(1 - x; p, w)
    \end{equation}%
\end{proposition}
\begin{proof}
    Let $\varphi_w$ and $\varphi_{w'}$ be the function $\varphi$, as defined in
    \cref{def:wowa}, for $w$ and $w'$, respectively. Also, let $\sigma$
    and $\sigma'$ be permutations of, respectively, $1 - x$ and $x$ in
    increasing order so that $\sigma'(i) = \sigma(m - i + 1)$. We have that:
    \begin{equation}
        \WOWA(1 - x; p, w) = 1 - \sum_{i = 1}^m\left[\varphi_{w\!}\!\left(\sum_{k
        = i}^m p_{\sigma(k)\!}\!\right) -
        \varphi_{w\!}\!\left(\smashoperator[r]{\sum_{k = i +
        1}^m}p_{\sigma(k)\!}\!\right)\right]x_{\sigma(i)}
    \end{equation}
    One can easily verify that $\varphi_{w'}(b) - \varphi_{w'}(a) =
    \varphi_w(1 - a) - \varphi_w(1 - b)$ for $a, b \in [0, 1]$ and that $\sum_{k =
    i}^m p_{\sigma(k)} = 1 - \sum_{k = 1}^{i - 1} p_{\sigma(k)}$. Thus:
    \begin{align}
        \WOWA(1 - x; p, w)
            &= 1 - \sum_{i = 1}^m \left[\varphi_{w'\!}\!\left(\sum_{k = 1}^{i}
            p_{\sigma(k)\!}\!\right) - \varphi_{w'\!}\!\left(\sum_{k = 1}^{i - 1}
            p_{\sigma(k)\!}\!\right)\right]x_{\sigma(i)} \\
            &= 1 - \sum_{i = 1}^m \left[\varphi_{w'\!}\!\left(\sum_{k = i}^m
            p_{\sigma'(k)\!}\!\right) - \varphi_{w'\!}\!\left(\smashoperator[r]{\sum_{k = i + 1}^{m}}
            p_{\sigma'(k)\!}\!\right)\right]x_{\sigma'(i)} \\
            &= 1 - \WOWA(x; p, w')\label{eq:reverse-prop-wowa}
    \end{align}
    \Cref{eq:prop-wowa} can then easily be obtained from
    \eqref{eq:reverse-prop-wowa}.%
    \qed
\end{proof}
For our overlap check, we wish to model a WOWA query as a \emph{linear ambit}
$\Ball[q, s; \Wt] = \{o : \Wt x_o\leq s\}$, where
$x_o=[d(q_i,o)]_i$, as introduced by \citet{Hetland:2019}.
While WOWAs are not linear functions, we can emulate a query with $m$ query
objects as a linear ambit with $m!$ facets, one per possible permutation of $x$.
Normally, the intersection check would require considering each facet in turn,
which would quickly become computationally unfeasible with an increasing $m$,
and could in theory lead to false positives.\footnote{This is discussed by
Hetland in Sect.~3.1~\citep{Hetland:2019}.} However, when the weights for the
WOWA representing our unfairness measure are in increasing order (corresponding
to a fairness measure on similarities, per \cref{prop:equiv}),
membership and overlap checks need only consider \emph{one} of
the facets, eliminating both of these problems.
\begin{proposition}\label{prop:ambit-check}
    Let $w$ and $p$ be weighting vectors, where $w_1 \le w_2 \le \dots \le w_m$.
    Let $W$ be a matrix with $m!$ rows, one for each possible permutation,
    $\sigma$, of an $m$-vector. For a permutation $\sigma$, the value in column
    $i$ of the corresponding row is:
    \begin{equation}
        \varphi\!\left(\sum_{k = j}^m p_{\sigma(k)\!}\!\right) -
        \varphi\!\left(\smashoperator[r]{\sum_{k
        = j + 1}^m} p_{\sigma(k)\!}\!\right)\;,
    \end{equation}
    \noindent where $\sigma(j) = i$ and $\varphi$ is the function from
    \cref{def:wowa}. For $x \in \smash{\mathbb{R}_{\ge 0}^m}$ and $s \in \mathbb{R}$,
    let $w_{\sigma}$ be the row in $\Wt$ \kern-2pt corresponding to $\sigma$. If
    $\sigma$ puts $x$ in increasing order, $\Wt x \le s$ iff $w_{\sigma}x \le s$.
    If $\sigma$ puts $x$ in decreasing order, $\Wt x > s$ iff $w_{\sigma}x > s$.
\end{proposition}
\begin{proof}
    For any permutation $\sigma$, we can create a new permutation $\sigma'$,
    with $\sigma'(i) = \sigma(i + 1)$, $\sigma'(i + 1) = \sigma(i)$ for some $i
    \in \{1, \dots, m - 1\}$ and $\sigma'(j) = \sigma(j)$ for all $j \notin \{i,
    i + 1\}$. Since $w$ is in increasing order, we know that the growth of
    $\varphi$ is monotonically decreasing over $[0, 1]$. Combined with the fact
    that $\|w_\sigma\|_1 = \varphi(1) - \varphi(0) = \|w\|_1 = 1$ for all
    $\sigma$, we get that:
    \begin{equation}\label{eq:sigma-cases}
        \begin{cases}
            w_\sigma x \ge w_{\sigma'}x & \text{if } x_{\sigma(i)} < x_{\sigma(i +
            1)} \\
            w_\sigma x \le w_{\sigma'}x & \text{if } x_{\sigma(i)} > x_{\sigma(i +
            1)} \\
            w_\sigma x = w_{\sigma'}x & \text{otherwise}
        \end{cases}
    \end{equation}
    If a permutation $\sigma$ does not put $x$ in increasing order, there
    is an $i$ such that $x_{\sigma(i)} > x_{\sigma(i + 1)}$. Thus, there
    is another permutation $\sigma'$ with $w_{\sigma'}x \ge w_{\sigma}x$.
    Consequently, one of the permutations $\sigma$ that maximizes
    $w_{\sigma}x$ must put $x$ in increasing order. Note that by the third
    case in \eqref{eq:sigma-cases}, if there are multiple permutations that
    put $x$ in increasing order, the value of $w_{\sigma}x$ is the same for all
    of them. Similarly, any $\sigma$ that puts $x$ in decreasing order minimizes
    $w_{\sigma}x$.
    \qed
\end{proof}
Using the construct in \cref{prop:ambit-check}, we can for a WOWA-based
unfairness measure, defined by weighting vectors $w$ and $p$, create a linear
ambit $\Ball[q, s; \Wt]$. As long as $w$ is in increasing order, i.e., $w =
w^\uparrow$, the membership check of this ambit, $\Wt x \le s$, is equivalent to
checking that $\WOWA(x; p, w^\uparrow) \le s$. That is, this ambit is equivalent
to a range query with the WOWA-based unfairness measure. And when checking
whether this query ambit intersects the inverted $r$-ball round $c$, we can in
principle perform $m!$ individual checks like $r - w_\sigma x < s$ (i.e.,
$w_\sigma x > r - s$),
one per row $\sigma$.\footnote{This follows from the linear ambit overlap check
described in Theorem~3.1.2 of \citet{Hetland:2019}, as well as from the
monotonicity result of \citet{Ciaccia:1998}, inserting the lower bound $r - x$
into the ambit membership predicate.}
\Cref{prop:ambit-check} shows us that we need only consider the single row
corresponding to a decreasing $x$. In other words, the overlap
check is strengthened from
$s > r - \WOWA(x; p, w^\uparrow)$
to
$s > r - \WOWA(x; p, w^\downarrow)$.

\section{Experiments}

To demonstrate the practical feasibility of the method, even without any
fine-tuning or high-effort optimization, we have tested it empirically on
synthetic and real-world data, using the basic index structure \emph{list of
clusters} (LC), as described by \citet{Chavez:2005}. Briefly, the LC partitions
the data set into a sequence of ball regions, each defined by a center, a
covering radius, and a set of member items. A search progresses by detecting
overlap with each ball in turn, potentially scanning its members for relevance.
A defining feature of LC is that the points in later buckets fall entirely
outside previous balls, so that if the query falls entirely \emph{inside} one of
the balls, the search process may be halted.

More specifically, bucket centers were chosen to maximize distance to previous
centers (heuristic \textit{p}5 of \citeauthor{Chavez:2005}), with each ball
constructed to contain the~\num{20} closest points to the center, as well as any
additional points that fall within the resulting radius.\footnote{These choices
were made based on the results of \citet{Chavez:2005}, which indicate that
\textit{p}5 yiels the best results overall, and a bucket size of~\num{20} is
a good tradeoff between filtering power and scanning time for a wide range of
data sets.}
The data sets used were:
\begin{itemize}
    \item Synthetic: \num{100000} uniformly random and clustered vectors from
        $[0,1]^{\mathrm D}$, for $\mathrm D=2,4,\dots,10$. The clustered vectors
        were constructed by first generating \num{1000} cluster centers,
        uniformly at random, and then generating \num{100} vectors per cluster,
        by adding standard Gaussian noise.
    \item Real-world: The Colors, NASA and Listeria SISAP data
        sets.\footnote{Available at \url{https://sisap.org}.}
\end{itemize}
Euclidean and Levenshtein distance were used with vectors and strings,
respectively. For the real-world data sets, the 101 first objects
were taken as queries; for the synthetic ones, queries were
generated in addition. These were used pairwise (1 and 2, 2 and 3, etc.)\@
in an OWA query with weights~1 and~3 (like the Gini coefficient). Fairest
neighbor queries ($k$FN) were run for $k=1,\dots,5$. The number of distance
computations was averaged over the~100 query object pairs.

\Cref{tab:results} shows the results. As a baseline, the number of distance
computations needed for a linear scan is listed, and the speedup for the
combined $k$FN query is shown for each $k$. For comparison, we also performed a
\emph{double} query, where a separate $k$ was found for each of the two query
objects, to ensure that the true $k$FN would be returned,\footnote{Of course,
these individual $k$ parameters would not be available when resolving a real
query, but this gives an optimistic bound for the competition.} and then two
separate $k$NN queries were performed, with the fairest neighbors found in their
intersection.
The combined search used fewer distance computations than the alternatives for
every data set and parameter setting. On average, the combined search used about
half as many distance computations as the separate queries, and a quarter of a
full linear scan.\footnote{More specifically, the average proportions, using the
geometric mean, were \SI{\propnaive}{\percent} and \SI{\propscan}{\percent},
respectively, corresponding to speedups of \num{\spdnaive} and \num{\spdscan}.}

\begin{table}[t]
\caption{Experimental results. For each of the double (two separate) and
combined queries, the speedup (where higher is better) from the number of
distance computations needed for linear scan is listed for each $k=1,\dots,5$.}
\label{tab:results}
\setlength{\tabcolsep}{.4em}
\begin{tabularx}{\linewidth}{@{}Xrrrrrrrrrrrr@{}}
\toprule
& & & \multicolumn{5}{c}{Double}
& \multicolumn{5}{c}{Combined}\\
\cmidrule(lr){4-8} \cmidrule(l){9-13}
Data set & Dim. & Scan & 1 & 2 & 3 & 4 & 5 & 1 & 2 & 3 & 4 & 5 \\
\midrule
Colors & 112
& \num{225162}
& 3.29
& 3.09
& 2.98
& 2.91
& 2.84
& 5.55
& 5.22
& 5.04
& 4.91
& 4.80
\\
NASA & 20
& \num{80098}
& 1.57
& 1.46
& 1.42
& 1.38
& 1.36
& 2.64
& 2.48
& 2.39
& 2.33
& 2.28
\\
Uniform & 4
& \num{200000}
& 2.17
& 2.14
& 2.12
& 2.10
& 2.09
& 7.13
& 6.94
& 6.82
& 6.72
& 6.65
\\
& 6
& \num{200000}
& 2.16
& 2.07
& 2.02
& 1.99
& 1.96
& 5.73
& 5.47
& 5.32
& 5.20
& 5.10
\\
& 8
& \num{200000}
& 2.01
& 1.89
& 1.82
& 1.78
& 1.74
& 4.20
& 3.95
& 3.79
& 3.68
& 3.59
\\
& 10
& \num{200000}
& 1.87
& 1.73
& 1.65
& 1.60
& 1.56
& 3.20
& 2.96
& 2.82
& 2.73
& 2.66
\\
Clustered & 4
& \num{200000}
& 2.28
& 2.25
& 2.22
& 2.21
& 2.19
& 7.62
& 7.47
& 7.37
& 7.30
& 7.23
\\
& 6
& \num{200000}
& 2.39
& 2.28
& 2.22
& 2.17
& 2.14
& 6.17
& 5.91
& 5.75
& 5.62
& 5.53
\\
& 8
& \num{200000}
& 2.07
& 1.93
& 1.87
& 1.82
& 1.79
& 4.41
& 4.12
& 3.95
& 3.84
& 3.75
\\
& 10
& \num{200000}
& 1.84
& 1.72
& 1.64
& 1.60
& 1.57
& 3.20
& 2.96
& 2.83
& 2.73
& 2.66
\\
Listeria & ---
& \num{41118}
& 1.25
& 1.17
& 1.16
& 1.13
& 1.06
& 1.28
& 1.28
& 1.27
& 1.27
& 1.27
\\
\bottomrule
\end{tabularx}

\end{table}

\section{Conclusions and Future Work}

Ordered weighted averages (OWA) and weighted OWAs (WOWA) may be used as a query
modality with any current metric index, when tradeoffs between multiple query
objects are needed, to find their $k$ fairest neighbors, $k$FN. They provide a
large degree of customizability, both in terms of their fairness profile and the
relative weights of different query objects, and are easy to implement. Other
monotone (un)fairness measures may also be used, though possibly with weakened
overlap checks in some cases.

Future research might look into adapting index structures, e.g., by adjusting
construction heuristics, to fair neighbor queries, and whether the requirements
for efficiency in practice are different from, say, single-object ball queries.
Generalizations of fairness might also be interesting, where one permits
negative weights for certain objects, which is straightforward for weighted sum,
but whose implementation is less obvious for WOWA. A more straightforward
extension of this work would be to test on other data sets, perhaps with higher
(intrinsic) dimensionality, using more advanced index structures.

\bibliographystyle{splncs04}
\bibliography{paper}

\end{document}